%% file: main.tex
\documentclass[copyright,creativecommons]{eptcs}
\providecommand{\event}{SSV 2012}

\newtheorem{theorem}{Theorem}[section]
\newtheorem{lemma}[theorem]{Lemma}

\newenvironment{proof}[1][Proof]{\begin{trivlist}
\item[\hskip \labelsep {\bfseries #1}]}{\end{trivlist}}
\newenvironment{definition}[1][Definition]{\begin{trivlist}
\item[\hskip \labelsep {\bfseries #1}]}{\end{trivlist}}

\newcommand{\qed}{\nobreak \ifvmode \relax \else
      \ifdim\lastskip<1.5em \hskip-\lastskip
      \hskip1.5em plus0em minus0.5em \fi \nobreak
      \vrule height0.75em width0.5em depth0.25em\fi}

\usepackage{comment}

\usepackage{subfig}
\usepackage{amsmath}
\usepackage{amssymb}
\usepackage{wrapfig}
\usepackage{stmaryrd}

\usepackage{multirow}   

\usepackage{pgf} 

\usepackage{tikz}
\usetikzlibrary{plotmarks}
\usetikzlibrary{automata}

\usepackage{xspace} 
\usepackage[ruled,linesnumbered,noend]{algorithm2e}
\DontPrintSemicolon
\usepackage{varioref} 

\usepackage{array}

\input{macros.tex}

\input{title}

\begin{document}
\maketitle
\input{abstract}

\input{introduction}

\input{timedautomata}

\input{timedart}

\input{experiments}

\input{conclusion}

\bibliographystyle{eptcs}
\bibliography{biblio}


\end{document}

%% file: macros.tex
\newcommand{\nat}{\ensuremath{\mathbb{N}}}
\newcommand{\natinf}{\ensuremath{\mathbb{N}^\infty}}

\newcommand{\config}{\ensuremath{\mathit{Conf}}}

\newcommand{\defin}{\stackrel{\rm def}{=}}

\newcommand{\techrep}[1]{ }

\newcommand{\dequi}{\ensuremath{\equiv_\maxConstant }}

\newcommand{\TA}{TA\xspace}
\newcommand{\uppaal}{{\sc Uppaal}\xspace}

\newcommand{\tapaal}{TAPAAL\xspace}

\newcommand{\dwl}{\ensuremath{\mathit{w}}}
\newcommand{\dpl}{\ensuremath{\mathit{p}}}

\newcommand{\TADefFull}{\ensuremath{A=\TADefTuple}\xspace}
\newcommand{\TADefTuple}{\ensuremath{(\TAlocs, \TAClocks, \TATrans,\TAinit)}\xspace}
\newcommand{\TAlocs}{\ensuremath{L}\xspace}
\newcommand{\TAClocks}{\ensuremath{C}\xspace}
\newcommand{\TATrans}{\ensuremath{\goes{}}\xspace}
\newcommand{\TAinit}{\ensuremath{\ell_0}\xspace}

\newcommand{\clockguards}{\ensuremath{{\cal G}}}

\newcommand{\anchor}{\ensuremath{\gamma\space}}
\newcommand{\setAnchor}{\ensuremath{\mathit{Anchors}\xspace}}

\newcommand{\maxOf}[1]{\ensuremath{\max(#1)\xspace}}
\newcommand{\minOf}[1]{\ensuremath{\min(#1)\xspace}}

\newcommand{\TDS}{\ensuremath{T_{\mathit{DS}}}}
\newcommand{\TBDS}{\ensuremath{T_{\mathit{BDS}}}}
\newcommand{\maxConstant}{\ensuremath{\mathit{MC}\xspace}}

\newcommand{\setIntervals}{\ensuremath{\mathcal{I}}\xspace}
\newcommand{\valuations}{\ensuremath{\mathcal{V}}\xspace}
\newcommand{\mcplus}{\ensuremath{\oplus}\xspace}

\newcommand{\punknown}{\ensuremath{\infty}\xspace}

\newcommand{\pwList}{\ensuremath{\mathit{PW}\xspace}}

\newcommand{\pwUdef}{\ensuremath{\bot\xspace}}

\newcommand{\lbound}{\ensuremath{\mathit{lb}\xspace}}
\newcommand{\ubound}{\ensuremath{\mathit{ub}\xspace}}
\newcommand{\lboundOf}[1]{\ensuremath{\mathit{lb}(#1)\xspace}}
\newcommand{\uboundOf}[1]{\ensuremath{\mathit{ub}(#1)\xspace}}

\newcommand{\goes}[1]{\ensuremath{\stackrel{#1}{\longrightarrow}}}

\newcommand{\goesDS}[1]{\ensuremath{\stackrel{#1}{\longrightarrow}_{DS}}}
\newcommand{\goesEDS}[1]{\ensuremath{\stackrel{#1}{\longrightarrow}_{BDS}}}

\newcommand{\dtts}{\ensuremath{\mathit{DTTS}}\xspace}

\newcommand{\addtopw}{\ensuremath{\mathtt{AddToPW}}\xspace}

%% file: title.tex
\newcommand{\kyrke}{Kenneth Y. J{\o}rgensen}
\newcommand{\kgl}{Kim G. Larsen}
\newcommand{\srba}{Ji\v{r}\'{\i} Srba}

\title{Time-Darts: A Data Structure for Verification of \\ Closed Timed
Automata\footnote{The paper was supported by VKR Center of Excellence MT-LAB.}}

\author{\kyrke \quad \quad \kgl \quad \quad \srba
\institute{Department of Computer Science, Aalborg University \\
             Selma Lagerl{\"o}fs Vej 300, 9220 Aalborg East, Denmark}
\email{\{kyrke,kgl,srba\}@cs.aau.dk}
}

\def\titlerunning{Time-Darts: A Data Structure for Verification of 
                 Closed Timed Automata}
\def\authorrunning{K.Y. J{\o}rgensen, K.G. Larsen \& J. Srba}

%% file: abstract.tex
\begin{abstract}
  Symbolic data structures  for model checking timed  systems have been
  subject  to a significant  research,  with Difference  Bound  Matrices
  (DBMs)  still being  the  preferred data  structure  in several  mature
  verification  tools.  In comparison,  discretization offers  an easy
  alternative,  with all  operations having  linear-time complexity  in the
  number of  clocks, and yet  valid for a large
  class of  \emph{closed}  systems.   Unfortunately,
  fine-grained discretization causes itself a state-space explosion.  
  We   introduce  a   new   data  structure   called
  \emph{time-darts} for the symbolic representation of state-spaces of
  timed  automata.  Compared  with the complete  discretization,  a single
  time-dart allows to represent an  arbitrary large set of states, yet
  the time complexity of operations on time-darts remain linear in the
  number  of clocks.
  We prove the correctness of the suggested reachability algorithm and
  perform several  experiments in order to compare  the performance of
  time-darts and the complete  discretization.
  The  main  conclusion  is   that in all our experiments 
  the  time-dart  method   
  outperforms the complete  discretization  and it scales 
  significantly better for models with larger constants.
\end{abstract}

%% file: introduction.tex
\section{Introduction}

Timed automata~\cite{AD94} are a well studied  formalism for modelling
and  verification  of real-time  systems.   Over  the years  extensive
research effort  has been made  towards the design of  data structures
and algorithms allowing for  efficient model checking of this modeling
formalism.  These techniques have by  now been implemented in a number
of    mature    tools     (e.g.     \uppaal    \cite{tutorial04},    IF
\cite{DBLP:conf/cav/BozgaGM02},      Kronos     \cite{DOTY96},     PAT
\cite{LiuSD08},     Rabbit     \cite{DBLP:conf/cav/BeyerLN03},     RED
\cite{DBLP:conf/rtcsa/HsiungWC00}),     with    zone-based    analysis
\cite{DBLP:conf/avmfss/Dill89,   DBLP:conf/ifip/BerthomieuM83}   still
being  predominant, stemming from  that  fact  that Difference  Bound
Matrices (DBMs)  offer a very  compact data structure  for efficient
implementation of the  various operations required for the state-space
exploration. 
Still the DBM data structure suffers  from the fact that all operations
have at least quadratic---and the crucial closure operation even 
cubic---time complexity in the number of clocks (though for diagonal-free
constraints the operations can be implemented in quadratic time~\cite{dbm-quadratic}).  
In   contrast, as  advocated in 
\cite{BMT:discrete:99,DBLP:conf/charme/Lamport05}, the use of \emph{discretization}
offers  an  easy  alternative,   with  all  operations  having  linear
complexity in the number of clocks, and yet valid for the large---and
in practice often sufficient---class of closed systems that contain only
nonstrict guards; moreover for reachability checking 
the continuous and discrete semantics coincide on this subclass.
%
%
%
%
%
\begin{figure}[t!]
\begin{minipage}[c]{0.5\linewidth}
\centering
\scalebox{1}{
\begin{tikzpicture}[font=\scriptsize]
\pgfnodecircle{p0}[stroke]{\pgfxy(0,4)}{7pt}{7pt}
\pgfnodecircle{p0s}[stroke]{\pgfxy(0,4)}{5pt}{7pt}
\pgfnodecircle{goal}[stroke]{\pgfxy(0,1.7)}{7pt}{7pt}

\pgfputat{\pgfxy(0.7,1.7)}{\pgfbox[center,center]{$\textit{Goal}$}}

\pgfsetendarrow{\pgfarrowsingle}
\pgfnodeconnline{p0}{goal}

\pgfnodeconncurve{p0}{p0}{240}{170}{1.5cm}{1.5cm}

\pgfnodeconncurve{p0}{p0}{140}{70}{1.5cm}{1.5cm}


\pgfnodeconncurve{p0}{p0}{300}{10}{1.5cm}{1.5cm}

\pgfputat{\pgflabel{.5}{\pgfxy(0,5)}{\pgfxy(0.6,-0.5)}{-4pt}}{\pgfbox[center,base]{

$x_1 = x_2 = \cdots = x_n = 0  \wedge y \geq 1$

}}

\pgfputat{\pgflabel{.5}{\pgfxy(1.5,10)}{\pgfxy(0.6,-1)}{-4pt}}{\pgfbox[center,base]{

$\cdots$

}}

\pgfputat{\pgflabel{.5}{\pgfxy(-3.5,11.5)}{\pgfxy(0.6,-1)}{-4pt}}{\pgfbox[center,base]{

$x_2 = 2, x_2:=0$

}}

\pgfputat{\pgflabel{.5}{\pgfxy(0,7)}{\pgfxy(0.6,-1)}{-4pt}}{\pgfbox[center,base]{
$x_1 = 1, x_1:=0$ $\;\;\;\;\;\;x_n = n, x_n:=0$
}}


\end{tikzpicture}
}

\end{minipage}
\begin{minipage}[c]{0.4\linewidth}
\centering
\small 
\begin{tabular}{|r|r|r|}
\hline
\textbf{Size}&\textbf{Discrete}&\textbf{Uppaal}\\ 
\hline
4&0.2& $<$0.1 \\ 
\hline
5&1.1 &2.2 \\  
\hline
6&5.5&27.5 \\ 
\hline
7&97& $>$300 \\ 
\hline
8&out of memory& - \\ 
\hline
\end{tabular}

\end{minipage}

\caption{Discrete vs. zone-based reachability algorithm (time in
  seconds)}
\label{fig:prime}
\end{figure}
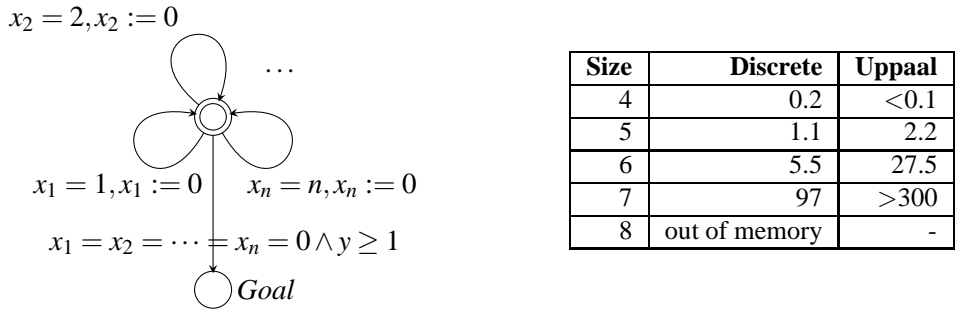

As  an  example consider  the  timed automaton  
shown  in Figure~\ref{fig:prime},
containing $n$  clocks and  $n$ self-loops  where the $i$'th  loop
has the guard $x_i=i$ and resets the clock $x_i$.  
We are interested in whether or not
we can  reach the \emph{Goal}  location.  For this to  happen, all
clocks  $x_1,\ldots,x_n$  must  \emph{simultaneously} have  the  value
zero,  corresponding  effectively to  calculating  the  least common
multiple of  the
numbers from $1$ to $n$. In Figure~\ref{fig:prime} we compare the
verification times of the zone-based reachability performed in \uppaal with that
of a simple Python  based implementation of discrete time reachability
checker for timed automata. 
Opposite to what one might expect, it turns
out that in this case the naive discrete  implementation without any
speed optimizations outperforms a state-of-the-art model checking tool.

On  the other  hand, the  disadvantage of  discretization is  that the
number  of states  to  be considered  explodes  when the  size of  the
constants  appearing in  the constraints  of the  timed  automaton are
increased.      In    fact,     the     experimental    results     of
Lamport~\cite{DBLP:conf/charme/Lamport05}
show that the zone-based methods outperform
discreterized methods when the maximum constant in the timed automaton
exceeds 10. Also  in \cite{DBLP:conf/charme/Lamport05}  the BDD-based
model checker SMV was  applied to symbolically represent the discreterized
state-space.  This  representation  is less  sensitive  to the
maximum constant of the model, yet in experimental 
results~\cite{DBLP:conf/forte/BeyerN03, DBLP:conf/hybrid/AsarinBKMPR97} it
appears that  the zone-based method  is still superior for  constants larger
than 16.

Inspired  by   the   success  of   discretization  reported   in
Figure~\ref{fig:prime},  we revisit the  problem of  finding efficient
data structures for the analysis of timed automata.  In particular, we
introduce  a  new  data  structure called  \emph{time-darts}  for  the
symbolic representation  of the state-spaces of  timed automata.  Compared
with  the  complete  discretization,  a  single  time-dart  allows us  to
represent an arbitrary large set of states, yet the time complexity of
operations  
remain  linear  in  the  number of  clocks,
providing a potential advantage compared to DBMs.

We propose a  symbolic reachability algorithm based on  a forward search.
To  ensure the termination of the forward search
the so-called  extrapolation  of  time   darts  with  respect  to the  maximum
constant appearing  in the model is required.   Given the subtleties
of extrapolation,\footnote{Despite  several earlier claims,  it was not
  before~\cite{DBLP:conf/stacs/Bouyer03} that  a complete---and a quite
  non-trivial---proof of correctness of zone-based forward reachability
  was given.}  we prove the termination  and correctness of  the proposed
algorithm.
We perform  several   experiments  in  order  to  compare  the
performance  of  time-darts versus the complete  discretization
representation.   The main  conclusion is  that the  time-dart method
consistently outperforms the complete  discretization and it is particularly
well suited for scaling up the constants used in the model.
Given the simplicity of implementing discrete-time algorithms compared
to the DBM-based ones, our method can be in practice well suited for the
verification of closed time systems with moderately large constants.

%% file: timedautomata.tex
\section{Timed Automata}


Let $\nat$ be the set of  nonnegative integers and let
$\natinf = \nat \cup \{ \infty \}$.
The comparison and addition operators
are defined as expected, in particular
$n < \infty$ and $n + \infty = \infty$ for $n \in \nat$.


A \emph{Discrete Timed Transition System} (\dtts) is a pair $T = (S,
\goes{})$ where
$S$ is a set of \emph{states}, and
$\goes{} \subseteq S \times (\nat \cup \{\tau\}) \times S$ is a
\emph{transition relation} written 
$s \goes{d} s'$ if $(s,d,s') \in \goes{}$ where $d \in \nat$ for \emph{delay actions}, and
$s \goes{\tau} s'$ if $(s,\tau,s') 
\in \goes{}$ for \emph{switch actions}.
By $\goes{}^*$ we denote the reflexive and transitive closure 
of the relation $\goes{} \defin \goes{\tau} \cup \bigcup_{d \in \nat} \goes{d}$.



Let \TAClocks be a finite set of clocks. 
A (discrete) \emph{clock valuation} of clocks from \TAClocks is a function $v:
\TAClocks \rightarrow \nat$. The set of all clock valuations
is denoted by $\valuations$.
Let $v \in \valuations$. 
We define the valuation $v+d$ after a delay
of $d \in \nat$ time units by 
$(v+d)(x) \defin v(x) + d$ for every $x \in \TAClocks$. 
For a subset $R \subseteq \TAClocks$ of clocks 
we define the valuation $v[R:=0]$ where all clocks from $R$ are
reset to zero  by
$v[R:=0](x) \defin v(x)$ 
for $x \in C \setminus R$ and $v[R:=0](x) \defin 0$ for $x \in R$.

A nonstrict (or closed) \emph{time interval} $I$ is of the
form $[a, b]$ or $[a, \infty)$ 
where $a, b \in \nat$ and $a \leq b$. 
The set of all time intervals is denoted by $\setIntervals$.
We use the functions $\ubound{}, \lbound{}: \setIntervals \goes{} \nat$ to 
return the upper resp. lower bound of a given interval.
A \emph{clock guard} over the set of clocks $\TAClocks$ is a function 
$g: \TAClocks \goes{} \setIntervals$ that assigns a time 
interval to each clock. 
We denote the set of all clock guards 
over $\TAClocks$ by $\clockguards(\TAClocks)$.
We write $v \models g$ for a valuation $v \in \valuations$
and a guard $g \in \clockguards(\TAClocks)$ whenever
$v(x) \in g(x)$ for all $x \in \TAClocks$. 

\begin{definition}[Timed Automaton]
A \emph{timed automaton} (TA) is a tuple $\TADefFull$ where 
$\TAlocs$ is a finite set of \emph{locations},
$\TAClocks$ is a finite set of \emph{clocks},
$\TATrans \subseteq L\times \clockguards(\TAClocks) \times 2^{\TAClocks}\times L$
is a finite \emph{transition relation} written
$\ell \goes{g,R} \ell'$ for $(\ell,g,R,\ell') \in \goes{}$, and
$\TAinit \in \TAlocs$ is an \emph{initial} location.
\end{definition} 
Note that we do not consider clock invariants as they can be substituted 
by adding corresponding clock guards to the outgoing transitions 
while preserving the answers to location-reachability checking.

A \emph{configuration} of a timed automaton $A$ 
is a pair $(\ell,v)$ where $\ell \in L$ 
and $v \in \valuations$. 
We denote the set of all configurations of $A$ by $\config(A)$.
The \emph{initial configuration} of $A$ is $(\ell_0,v_0)$ where 
$v_0(x) \defin 0$ for all $x \in \TAClocks$. 



\begin{definition}[Discrete Semantics]
A TA $\TADefFull$ generates a \dtts $\TDS(A) \defin (\config(A), \goesDS{})$ 
where states are configurations of $A$ and the transitions are
given by
\begin{center}
  \begin{tabular}{ll}
   $(\ell,v) \goesDS{\tau} (\ell',v[R:=0])$ & if
        $\ell \goes{g,R} \ell'$ such that $v \models g$ 
  \\
   $(\ell,v) \goesDS{d} (\ell,v+d)$ & if 
	$d \in \nat$.
  \end{tabular}
\end{center}
\end{definition}


The discrete semantics clearly yields an infinite state space
due to unbounded time delays. We will
now recall that the reachability problem for a \TA $A$ 
can be solved by looking only 
at a finite prefix of the state space up to some constant
determining the largest possible delay.
Let $\maxConstant$ be the largest integer that appears
in any guard of $A$.
Two valuations $v, v' \in \valuations$ are \emph{equivalent}
up to the maximal constant $\maxConstant$, written $v \dequi v'$, if 
$$\forall x \in \TAClocks.\ v(x) = v'(x) \; \vee \; 
(v(x) > \maxConstant \; \wedge \; v'(x) > \maxConstant).
$$

Observe that the equivalence relation $\dequi$ has only
finitely many equivalence classes as there are finitely many
clocks and each of them is bounded by the constant $\maxConstant$. 
\begin{lemma} \label{lem:sameguards}
Let $v,v' \in \valuations$ s.t. $v \dequi v'$ and
let  $g\in \clockguards(\TAClocks)$ be a guard where 
$0 \leq \lbound(g(x)) \leq \maxConstant$, and 
$\ubound(g(x)) = \infty$ or $0 \leq \ubound(g(x)) \leq \maxConstant$ for all $x \in \TAClocks$.
Then $v \models g$ iff  $v' \models g$.
\end{lemma}

Moreover, any two configurations
with the same location and equivalent valuations are
timed bisimilar (for the definition of timed bisimilarity
see e.g~\cite{LY:97}). 


\begin{lemma} \label{lem:bisim}
The relation $B= \{ ( (\ell,v),(\ell,v') \mid v \dequi v' \}$ 
is a timed bisimulation for any timed automaton with
its maximum constant $\maxConstant$.
\end{lemma}
\begin{proof}
Let $((\ell,v),(\ell,v')) \in B$. We analyse only
the switch and delay actions from $(\ell,v)$; the
situation for the transitions from $(\ell,v')$ is symmetric.
\begin{itemize}
\item Assume that $(\ell,v) \goesDS{\tau} (\ell',v[R:=0])$ via a transition
$\ell \goes{g,R} \ell'$.
Due to Lemma~\ref{lem:sameguards} and the fact that
$v \models g$, we get $v' \models g$. Hence also
$(\ell,v') \goesDS{\tau} (\ell',v'[R:=0])$ and it is easy to verify that
$v[R:=0] \dequi v'[R:=0]$.
\item Assume that $(\ell,v) \goesDS{d} (\ell,v+d)$.
We want to argue that also $(\ell,v') \goesDS{d} (\ell,v'+d)$ such that
$v+d \dequi v'+d$, however, this is easy to see from that facts
that (i) if $v(x),v'(x) > \maxConstant$ then also
$(v+d)(x),(v'+d)(x) > \maxConstant$ and (ii)
if $v(x)=v'(x) \leq \maxConstant$ then
$(v+d)(x)=(v'+d)(x)$.
\qed
\end{itemize}
\end{proof}

We now define an alternative discrete semantics
of \TA with only finitely many reachable configurations. 
First, for the maximum constant $\maxConstant$, we define
a bounded addition operator 
$$ n \mcplus_\maxConstant m \defin \begin{cases} \maxConstant + 1 & 
\text{if $n+m > \maxConstant$,} \\
                            n+m &\text{otherwise.}
\end{cases}$$
The operation $\mcplus_\maxConstant$ is in a natural way extended 
to functions and tuples.

\begin{definition}[Bounded Discrete Semantics]
A TA $\TADefFull$ with the maximal constant 
$\maxConstant$ generates a 
\dtts
$\TBDS(A) \defin (\config(A), \goesEDS{})$ where states are configurations of $A$ and the transition relation
$\goes{}$ is defined by
\begin{center}
  \begin{tabular}{ll}
   $(\ell,v) \goesEDS{\tau} (\ell',v[R:=0])$ & if
        $\ell \goes{g,R} \ell'$ such that $v \models g$ 
\\
   $(\ell,v) \goesEDS{d} (\ell,v \mcplus_{\maxConstant} d)$ & if 
	$d \in \nat$.
  \end{tabular}
\end{center}
\end{definition}


We say that a location $\ell_g$ is \emph{reachable} 
in $\TDS(A)$ resp.  in $\TBDS(A)$ if 
$(\ell_0, v_0) \goes{}^* (\ell_g,v)$ for some valuation $v$
where $\goes{}$ is $\goesDS{}$ resp. $\goesEDS{}$.

We conclude that the bounded semantics preserves reachability
of locations, the main problem we are interested in.
This fact follows from Lemma~\ref{lem:bisim}.

\begin{theorem} \label{thm:timedreachabiligy}
A location $\ell$ is reachable in $\TDS(A)$ iff
$\ell$ is reachable in $\TBDS(A)$.
\end{theorem}


\section{Naive Reachability Algorithm}

We can now describe the naive search algorithm
that explores in a standard way, point by point, the finite state-space of the
bounded semantics and provides the answer
to the location reachability problem. 
Algorithm~\ref{alg:naivealg} searches through all reachable states, starting 
from the initial location, until 
a goal configuration is found (returning true) 
or all configurations are visited (returning false). 
Notice that the algorithm is nondeterministic
as it is not specified what element should be removed from
$\textit{Waiting}$ at line~\ref{alg:naivealg:remove} (such choice
depends on the concrete search strategy like DFS or BFS). 
The next theorem states that Algorithm~\ref{alg:naivealg} is correct.

\begin{algorithm}[t]
  \KwIn{A timed automaton \TADefFull and a location $\ell_g \in \TAlocs$}
  \KwOut{true if $\ell_g$ is reachable in $\TDS(A)$, false otherwise}
  \SetKwFunction{AddToPW}{AddToPW}
  \Begin{

	$\mathit{Passed}:= \emptyset;\ \mathit{Waiting} := \emptyset; $\; 
	\AddToPW{$\ell_0, v_0$} \label{alg:naivealg:firstadd} \;

	\While{Waiting $ \neq \emptyset$ \label{alg:naivealg:mainloop}}  {
		remove some $(\ell, v)$ from $\mathit{Waiting}$ \label{alg:naivealg:remove} \;
		$\mathit{Passed} := \mathit{Passed} \cup \{ (\ell, v) \} $
\label{alg:naivealg:addtopass} \;

		\ForAll {$(\ell', v')$ such that $(\ell, v) \goesEDS{\tau} (\ell',v')$
\label{alg:naivealg:forall}} {
			\AddToPW{$\ell', v'$} \label{alg:naivealg:adddis} \;
			}
			\AddToPW{$\ell, (v\mcplus_{\mathit{MC}} 1) $}  \label{alg:naivealg:addtime} \;

	}
	\Return false \label{alg:naivealg:false}\;
	}

	\vspace{.4cm}

	\AddToPW{$\ell$, $v$}\; {
	{
		\If{ $(\ell, v) \notin \mathit{Passed} \; \cup
\textit{Waiting}$ \label{alg:naivealg:testpw} }{

			\If {$ \ell = \ell_g $ \label{alg:naivealg:test}}{
				\Return true \label{alg:naivealg:true}\ \ \ \  /* and terminate the whole algorithm */\;
			}\Else {
				$\textit{Waiting} := \textit{Waiting} \cup \{
(\ell, v) \}$ \label{alg:naivealg:addwaiting} \;
			}

		} 	
}

  }
  \caption{Naive reachability algorithm}
  \label{alg:naivealg}
\end{algorithm}

\begin{theorem} \label{thm:discreteAlgorithm}
Let $A$ be a timed automaton and let $\ell_g$ be a location.
Algorithm~\ref{alg:naivealg} terminates, and it returns true iff
$\ell_g$ is reachable in the discrete semantics $\TDS(A)$.
\end{theorem}
\begin{proof}
First notice that the algorithm terminates because
there is only a finite number of configurations
that can be possibly added to $\textit{Waiting}$:    
the number of locations is finite and due to the bounded addition
at line~\ref{alg:naivealg:addtime} the total number of configurations
is finite too.
Whenever a configuration
is removed from the set $\textit{Waiting}$, it is added to
$\textit{Passed}$ (line \ref{alg:naivealg:addtopass})
and can never be inserted into $\textit{Waiting}$ again
due to the test at line \ref{alg:naivealg:testpw}.
As we remove one element from $\textit{Waiting}$ each time
the body of the while-loop is executed, the algorithm necessarily
terminates either at line~\ref{alg:naivealg:false} or even
earlier at line~\ref{alg:naivealg:true} if the
goal location is reachable.

Now we prove the correctness part.
By Theorem~\ref{thm:timedreachabiligy} we can equivalently argue
that the algorithm returns true iff $\ell_g$ is reachable
in $\TBDS(A)$.	

``$\Rightarrow$'': Assume that Algorithm~\ref{alg:naivealg} returns true.
We want to show that the location $\ell_g$ is reachable in $\TBDS(A)$. 
This can be established by the following invariant:
any call of $\addtopw$ with the argument $(\ell, v)$ implies
that the configuration $(\ell,v)$ is reachable in $\TBDS(A)$.
For the initialisation at line~\ref{alg:naivealg:firstadd}
this clearly holds. In the while-loop the calls to $\addtopw$
are at lines~\ref{alg:naivealg:adddis} and~\ref{alg:naivealg:addtime}.
At line~\ref{alg:naivealg:adddis} we know by the invariant
that $(\ell,v)$ is reachable and we call $\addtopw$ only
with $(\ell',v')$ such that $(\ell,v) \goesEDS{\tau} (\ell',v')$, 
so the invariant is preserved. Similarly
at line~\ref{alg:naivealg:addtime} for the argument
$(\ell, (v\mcplus_{\textit{MC}} 1))$ of $\addtopw$ holds that
$(\ell,v) \goesEDS{1} (\ell, (v\mcplus_{\textit{MC}} 1))$, so it
is reachable as well. 
 
``$\Leftarrow$'': Assume that a configuration $(\ell',v')$
is reachable via $n$ transitions in $\TBDS(A)$, formally
$(\ell_0, v_0) \goesEDS{}^n (\ell', v')$, where (without loss
of generality) all delay transitions in the sequence are of the
form $\goesEDS{1}$, in other words they add exactly one-unit time delay. 
By induction on $n$ we will establish that during
any execution of the algorithm
there is eventually a call of $\addtopw$ with the argument $(\ell',v')$,
unless the algorithm already returned true.
If $n=0$ then the claim is trivial due to the call at 
line~\ref{alg:naivealg:firstadd}. If $n>0$ then
either (i) $(\ell_0, v_0) \goesEDS{}^{n-1} (\ell,v) \goesEDS{\tau} (\ell', v')$
with the last switch action 
or (ii) $(\ell_0, v_0) \goesEDS{}^{n-1} (\ell,v) \goesEDS{1} (\ell', v')$
with the last one-unit delay action. By induction hypothesis,
unless the algorithm already returned true, there will be eventually
a call of $\addtopw$ with the argument $(\ell,v)$ and this element
is added to the set $\textit{Waiting}$. Because the algorithm
terminates, the element $(\ell,v)$ will be
eventually removed from $\textit{Waiting}$ at line~\ref{alg:naivealg:remove}
and its switch successors and the one-unit delay successor
will become arguments of the call to $\addtopw$ at 
lines~\ref{alg:naivealg:adddis} and~\ref{alg:naivealg:addtime}. Hence
the induction hypothesis for the cases (i) and (ii) is established.
\qed
\end{proof}

%% file: timedart.tex
\section{Time-Dart Data Structure} \label{sec:darts}

We shall now present a novel symbolic representation of the 
discrete state space.
The symbolic structure, we call it a \emph{time-dart}, allows us to
represent a number of concrete configurations in a more compact way
so that time successors of a configuration are stored without
being explicitly enumerated. 
We start with the definition of an \emph{anchor point},
denoting the beginning of a time-dart.

\begin{figure}[t]
\begin{center}
\begin{tikzpicture}[scale=0.50]
\draw[help lines] (0,0) grid (9,9);
\draw[->,thick] (0,0) -- (10,0);
\draw[->,thick] (0,0) -- (0,10); 
\node at (5.0,10) {time-dart $((2,0),2,5)$};
\node at (-.6,10.0) {\scriptsize \rotatebox{90}{clock $y$}};
\node at (10.0,-.5) {\scriptsize clock $x$};
\draw[dashed] (2,0) -- (12,10);
\filldraw[fill=black, draw=black] (1.8,-0.2) -- (2.2,-0.2) -- (2.2,0.2) -- (1.8,0.2) -- (1.8,-0.2);
\draw[dotted] (2,0) -- (2.6,-0.5); \node at (2.7,-0.6) {$\anchor$};
\draw[dotted] (4,2) -- (4.6,1.5); \node at (4.7,1.4) {$w$};
\draw[dotted] (7,5) -- (7.6,4.5); \node at (7.7,4.4) {$p$};

\filldraw[fill=white,draw=black, thick] (4,2) circle(6pt);
\filldraw[fill=white,draw=black, thick] (5,3) circle(6pt);
\filldraw[fill=white,draw=black, thick] (6,4) circle(6pt);
\filldraw[black] (7,5) circle(6pt);
\filldraw[black] (8,6) circle(6pt);
\filldraw[black] (9,7) circle(6pt);
\filldraw[black] (10,8) circle(6pt);
\filldraw[black] (11,9) circle(6pt);
\end{tikzpicture}
\end{center}
\caption{A time-dart $(\anchor,w,p)$ where $\anchor(x)=2$, $\anchor(y)=0$, 
$w=2$ and $p=5$}
\label{fig:dart}
\end{figure}

An \emph{anchor point} over a set of clocks $\TAClocks$ is
a  clock valuation $\anchor: \TAClocks \goes{} \nat$
where $\anchor(x) = 0 $ for at least one
$x\in \TAClocks$.
We denote the set of all anchor points over a set of clocks $\TAClocks$
by $\setAnchor(\TAClocks)$.
Now we are ready to define
\emph{time-darts}.

\begin{definition}[Time-Dart]
A \emph{time-dart} over a set of clocks \TAClocks
is a triple 
$(\anchor,w,p)$ where
$\anchor \in \setAnchor(\TAClocks)$ is an anchor point,
$w \in \nat$ is a waiting distance, and
$p \in \natinf$ is a passed distance such that $w \leq p$.
\end{definition}

The intuition is that a 
time-dart describes the corresponding passed and waiting sets
in a given location. Figure~\ref{fig:dart} shows a dart example with
two clocks $x$ and $y$, anchor point $(2,0)$, waiting distance $2$
and passed distance $5$. The empty circles represent the points 
in the waiting set and the filled circles represent the points in
the passed set, formally defined by: 
$\mathit{Waiting}(\anchor, w, p)  = \{(\anchor + d) \mid \dwl \leq d < \dpl\}$
and
$\mathit{Passed}(\anchor, w, p)  = \{ (\anchor + d) \mid  d \geq \dpl \}$. 

The \emph{passed-waiting list} is represented as a function from
locations and anchor points to the corresponding waiting and passed
distances (here $\pwUdef$ represents the undefined value): 
$$\pwList: L \times \setAnchor 
\goes{} (\nat \times \natinf)  \cup \{ \pwUdef \}\ .$$
Such a structure can be conveniently implemented as
a hash map.
A given passed-waiting list $\pwList$ 
defines the sets of passed and waiting configurations. 
\begin{align*}
\mathit{Waiting}(\pwList)  & = \{ (\ell, v) \mid \exists \anchor. \pwList(\ell,
\anchor) = (w,p) \neq  \pwUdef 
\text{ and } v \in \mathit{Waiting}(\anchor, w, p)  \}     \\
\mathit{Passed}(\pwList)  & = \{ (\ell, v) \mid \exists \anchor. \pwList(\ell,
\anchor) = (w,p) \neq  \pwUdef 
\text{ and } v \in \mathit{Passed}(\anchor, w, p) \}      
\end{align*} 


\section{Reachability Algorithm Based on Time-Darts}
\begin{algorithm}[t!]
  \KwIn{A timed automaton $\TADefFull$ and a location $\ell_g \in \TAlocs$}
  \KwOut{true if $\ell_g$ is reachable in $\TDS(A)$, false otherwise}

  \SetKwFunction{AddToPW}{AddToPW}
  \SetKwFunction{Break}{break}

\vspace{.3cm}
  \Begin{

	$\pwList(\ell, \anchor) := \pwUdef $ for all $(\ell, \anchor)$\ \ \ /* default value */\;
	\AddToPW{$\ell_0, \anchor_0,0,\punknown$} where $\anchor_0(x):=0$ for all $x \in \TAClocks$ \label{alg:addInit} \;

	\While{ $\exists (\ell, \anchor). \; \pwList(\ell, \anchor) = (w,p) \text{ and } w < p$} { \label{alg:whileChoice}
	$\pwList(\ell, \anchor) := (w, w) \label{alg:addSelf}$\;
	\ForEach{ $(\ell, g, R, \ell') \in \goes{}$ }{ \label{alg:foreach}
		$ \mathit{start} := \maxOf{w,\maxOf{ \{ \lboundOf{g(x)} - \anchor(x) 
  \mid x \in C \}}}$ \label{alg:start}\;
		$ \mathit{end} := \minOf{ \{ \uboundOf{g(x)} - \anchor(x) \mid x \in
C \}}$ \label{alg:end} \;

		\If{$(\mathit{start} < p \; \wedge \; \mathit{start} \leq 
\mathit{end})$
\label{alg:guardsat}}{ 
		  \If{$R = \emptyset$}{
\label{alg:resetStart}
			\AddToPW{{$\ell', (\anchor \mcplus_{\maxConstant} \mathit{start}) - \mathit{start}, \mathit{start}, \punknown$}} \label{alg:addOneSuccToPW}\;
		  }\Else{ 
$\mathit{stop} := \max\{{\mathit{start},
\maxConstant+1-\min_{x \in C\smallsetminus R}{\anchor(x)}}\}$ \;
			\For{$n := \mathit{start} \text{ to } \minOf{\mathit{end}, p-1,\mathit{stop}$} }{
				\AddToPW{{$\ell', (\anchor
\mcplus_{\maxConstant} n)[R:=0]  , 0, \punknown$}}
\label{alg:addManySuccToPW}\; \label{alg:resetStop}
			}
		}

		}
	}
	}
	\Return false

  }

\vspace{.4cm}

\AddToPW{$\ell$, $\anchor$, $w$, $p$}\; {
	{
	\If{$\ell=\ell_g$}{
		\Return true \ \ \ \  /* and terminate the whole algorithm */
\label{alg:returnTrue}
	}
		\If{ $\pwList(\ell, \anchor) = \pwUdef $ }{

			 $\pwList(\ell, \anchor) := (w, p)$ \;

		} \Else {

			$(w', p') := \pwList(\ell, \anchor)$ \;
				$\pwList(\ell, \anchor) := (\minOf{w, w'},
\minOf{p, p'})$ \label{alg:addMinToPW} \;
		}

	}
}

  \caption{Time-dart reachability algorithm}
  \label{alg:reachability}
\end{algorithm}

\input{fig-successor}

We can now present Algorithm~\ref{alg:reachability} showing us how
time-darts can be used to compute the set of reachable states
of a timed automaton in a compact and efficient way.
The algorithm repeatedly selects from the waiting list a location $\ell$ with a 
time-dart $(\anchor,w,p)$ 
that still contains some unexplored points ($w<p$). Then for each
edge $\ell \goes{g,R} \ell'$ in the timed automaton
it computes the $\mathit{start}$ and $\mathit{end}$ delays
from the anchor point such that $\mathit{start}$ is the minimum
delay where the guard $g$ gets first enabled and
$\mathit{end}$ is the maximum possible delay so that $g$ is still enabled.
Depending on the concrete situation it will add a new
time-dart (or a set of darts) with location $\ell'$ to the waiting list
by calling $\addtopw$. A switch transition is always followed 
by a delay transition 
that is computed symbolically (including in a single step all possible delays).
There are several cases that determine what kinds of
new time-darts are generated. Figure~\ref{fig:successor} gives
a graphical overview of the different situations.
In Figure~\ref{fig:succ1} we illustrate the produced time-dart that serves 
as the argument for the call
to $\addtopw$ at line~\ref{alg:addOneSuccToPW} of the algorithm
(no clocks are reset).
Here the anchor point $\anchor$ is not modified because 
$((\anchor \oplus_\maxConstant \mathit{start}) - \mathit{start}) = \anchor$.
Figure~\ref{fig:succ2} shows another example of a call at
line~\ref{alg:addOneSuccToPW} where the anchor point
changes. Finally, Figure~\ref{fig:succ3} explains the case
where some clocks are reset and several new darts are added
in the body of the for-loop at line~\ref{alg:addManySuccToPW} of the algorithm
(the for-loop starts from $\mathit{start}$ and stops as soon as either
$\mathit{end}$, the beginning of the passed list, or the number 
$\mathit{stop}$---used for performance optimization--is reached).
We note that the figures show the time-darts that the function
$\addtopw$ is called with; inside the function the information
already stored in the passed-waiting list for the concrete anchor point
and location
is updated so that we take the minimum of the current and new
waiting and passed distances (line~\ref{alg:addMinToPW} of the algorithm).

\input{alg-example}

The correctness theorem requires a detailed technical
treatment and its complete proof is given in the full version of this paper.
Termination follows from the fact that newly added anchor points
are computed as $(\anchor \oplus_\maxConstant \mathit{start}) - \mathit{start}$
or $(\anchor \oplus_\maxConstant n)[R:=0]$ which ensures a finite
size of the passed-waiting list and that every time-dart $(\anchor,w,p)$ 
on the list satisfies $0 \leq w \leq \maxConstant$, $w < p$, and 
$p \leq \maxConstant$ or $p= \infty$. 
Soundness proof is by
a case analysis establishing a loop-invariant that
every call to $\addtopw$ only adds time-darts that represent 
reachable configurations in the bounded semantics. 
Finally, the completeness
proof is done by induction on the length of the computation leading to
a reachable configuration, taking into account the nondeterministic
nature of the algorithm, the fact that $\dequi$ is a timed bisimulation,
and it makes a full analysis of the different cases for adding new time-darts
present in the algorithm for its performance optimization.

\begin{theorem} \label{thm:timedartAlgorithm}
Let $A$ be a timed automaton and let $\ell_g$ be a location.
Algorithm~\ref{alg:reachability} terminates, and it returns true iff
$\ell_g$ is reachable in the discrete semantics $\TDS(A)$.
\end{theorem}

%% file: fig-successor.tex
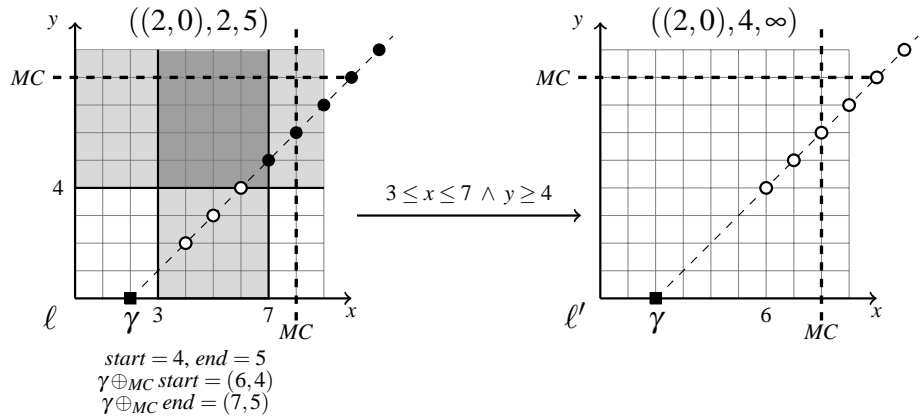
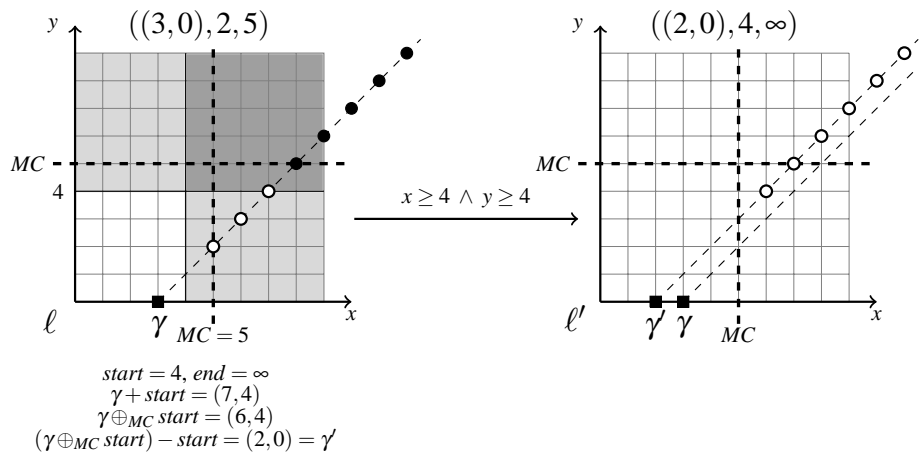
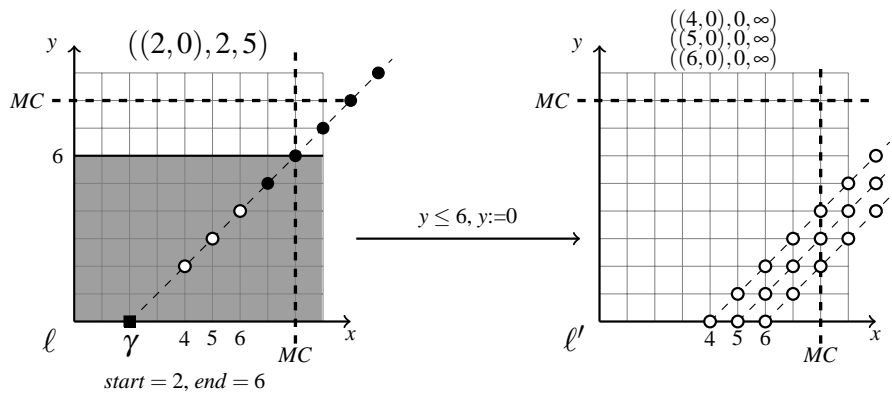
\begin{figure}[th!]
\begin{center}
\subfloat[Unchanged anchor point]{
\scalebox{1.05}{
\begin{tikzpicture}[scale=0.35]
\filldraw[fill=gray!30, draw=white, thick]
 (3,0) -- (7,0) -- (7,9) -- (3,9) -- (3,0);
\filldraw[fill=gray!30, draw=white, thick]
 (0,4) -- (9,4) -- (9,9) -- (0,9) -- (0,4);
\filldraw[fill=gray!75, draw=white, thick]
    (3,4) -- (7,4) -- (7,9) -- (3,9) -- (3,4);

\draw[help lines] (0,0) grid (9,9);
\draw[->,thick] (0,0) -- (10,0);
\draw[->,thick] (0,0) -- (0,10);
\node at (4.5,10) {$((2,0),2,5)$};
\draw[very thick, dashed] (8,-0.8) -- (8,9.9);
\draw[very thick, dashed] (-0.8,8) -- (9.9,8);
\node at (-1.7,8) {\mbox{\scriptsize $\maxConstant$}};
\node at (8,-1.2) {\mbox{\scriptsize $\maxConstant$}};
\node at (-.8,10.0) {\scriptsize \rotatebox{0}{$y$}};
\node at (10.0,-.5) {\scriptsize $x$};
\draw[dashed] (2,0) -- (11.5,9.5);
\filldraw[fill=black, draw=black] (1.8,-0.2) -- (2.2,-0.2) -- (2.2,0.2) -- (1.8,0.2) -- (1.8,-0.2);
\node at (2,-0.9) {$\anchor$};
\node at (3,-0.6) {\scriptsize $3$};
\node at (7,-0.6) {\scriptsize $7$};
\node at (-0.6,4) {\scriptsize $4$};
\node at (-0.9,-0.7) {\mbox{\large $\ell$}};

\node at (4,-2.2) {\mbox{\scriptsize $\mathit{start}=4$, $\mathit{end}=5$}};
\node at (4,-3.0) {\mbox{\scriptsize $\anchor \oplus_\maxConstant \mathit{start} = (6,4)$}};
\node at (4,-3.8) {\mbox{\scriptsize $\anchor \oplus_\maxConstant \mathit{end} = (7,5)$}};

\draw[thick] (3,9) -- (3,0) -- (7,0) -- (7,9);
\draw[thick] (0,9) -- (0,4) -- (9,4);

\filldraw[fill=white,draw=black, thick] (4,2) circle(6pt);
\filldraw[fill=white,draw=black, thick] (5,3) circle(6pt);
\filldraw[fill=white,draw=black, thick] (6,4) circle(6pt);
\filldraw[black] (7,5) circle(6pt);
\filldraw[black] (8,6) circle(6pt);
\filldraw[black] (9,7) circle(6pt);
\filldraw[black] (10,8) circle(6pt);
\filldraw[black] (11,9) circle(6pt);

\begin{scope}[xshift=9.2cm,yshift=-2cm]
\draw[->,thick] (1,5) to node[above] 
     {\scriptsize $3 \leq x \leq 7\ \wedge\ y \geq 4$} (9.1,5);
\end{scope}

\begin{scope}[xshift=19cm]
\draw[help lines] (0,0) grid (9,9);
\draw[->,thick] (0,0) -- (10,0);
\draw[->,thick] (0,0) -- (0,10);
\node at (4.5,10) {$((2,0),4,\infty)$};
\draw[very thick, dashed] (8,-0.8) -- (8,9.9);
\draw[very thick, dashed] (-0.8,8) -- (9.9,8);
\node at (-1.7,8) {\mbox{\scriptsize $\maxConstant$}};
\node at (8,-1.2) {\mbox{\scriptsize $\maxConstant$}};
\node at (-.8,10.0) {\scriptsize \rotatebox{0}{$y$}};
\node at (10.0,-.5) {\scriptsize $x$};
\draw[dashed] (2,0) -- (11.5,9.5);
\filldraw[fill=black, draw=black] (1.8,-0.2) -- (2.2,-0.2) -- (2.2,0.2) -- (1.8,0.2) -- (1.8,-0.2);
\node at (2,-0.9) {$\anchor$};
\node at (6,-0.6) {\scriptsize $6$};
\node at (-0.9,-0.6) {\mbox{\large $\ell'$}};

\filldraw[fill=white,draw=black, thick] (6,4) circle(6pt);
\filldraw[fill=white,draw=black, thick] (7,5) circle(6pt);
\filldraw[fill=white,draw=black, thick] (8,6) circle(6pt);
\filldraw[fill=white,draw=black, thick] (9,7) circle(6pt);
\filldraw[fill=white,draw=black, thick] (10,8) circle(6pt);
\filldraw[fill=white,draw=black, thick] (11,9) circle(6pt);
\end{scope}
\end{tikzpicture}
\label{fig:succ1}
} }

\subfloat[Shift of anchor point]{
\scalebox{1.05}{
\begin{tikzpicture}[scale=0.35]
\filldraw[fill=gray!30, draw=gray!20, thick]
 (4,0) -- (9,0) -- (9,9) -- (4,9) -- (4,0);
\filldraw[fill=gray!30, draw=gray!20, thick]
 (0,4) -- (9,4) -- (9,9) -- (0,9) -- (0,4);
\filldraw[fill=gray!75, draw=gray!20, thick]
    (4,4) -- (9,4) -- (9,9) -- (4,9) -- (4,4);

\node at (-0.6,4) {\scriptsize $4$};

\draw[help lines] (0,0) grid (9,9);
\draw[->,thick] (0,0) -- (10,0);
\draw[->,thick] (0,0) -- (0,10);
\node at (4.5,10) {$((3,0),2,5)$};
\draw[very thick, dashed] (5,-0.8) -- (5,9.3);
\draw[very thick, dashed] (-0.8,5) -- (9.9,5);
\node at (-1.7,5) {\mbox{\scriptsize $\maxConstant$}};
\node at (5,-1.2) {\mbox{\scriptsize $\maxConstant=5$}};
\node at (-.8,10.0) {\scriptsize \rotatebox{0}{$y$}};
\node at (10.0,-.5) {\scriptsize $x$};
\draw[dashed] (3,0) -- (12.5,9.5);
\filldraw[fill=black, draw=black] (2.8,-0.2) -- (3.2,-0.2) -- (3.2,0.2) -- (2.8,0.2) -- (2.8,-0.2);
\node at (3,-0.9) {$\anchor$};
\node at (-0.9,-0.7) {\mbox{\large $\ell$}};

\node at (4,-2.6)
   {\mbox{\scriptsize $\mathit{start}=4$, $\mathit{end}=\infty$}};
\node at (4,-3.4) {\mbox{\scriptsize $\anchor + \mathit{start} = (7,4)$}};
\node at (4,-4.2) 
{\mbox{\scriptsize $\anchor \oplus_\maxConstant \mathit{start} = (6,4)$}};
\node at (4,-5) 
{\mbox{\scriptsize $(\anchor \oplus_\maxConstant \mathit{start}) -
\mathit{start} = (2,0) = \anchor'$}};

\draw (4,0) -- (4,9);
\draw (0,4) -- (9,4);

\filldraw[fill=white,draw=black, thick] (5,2) circle(6pt);
\filldraw[fill=white,draw=black, thick] (6,3) circle(6pt);
\filldraw[fill=white,draw=black, thick] (7,4) circle(6pt);
\filldraw[black] (8,5) circle(6pt);
\filldraw[black] (9,6) circle(6pt);
\filldraw[black] (10,7) circle(6pt);
\filldraw[black] (11,8) circle(6pt);
\filldraw[black] (12,9) circle(6pt);

\begin{scope}[xshift=9.1cm,yshift=-2cm]
\draw[->,thick] (1,5) to node[above] 
     {\scriptsize $x \geq 4 \ \wedge\ y \geq 4$} (9.1,5);
\end{scope}

\begin{scope}[xshift=19cm]
\draw[help lines] (0,0) grid (9,9);
\draw[->,thick] (0,0) -- (10,0);
\draw[->,thick] (0,0) -- (0,10);
\node at (4.5,10) {$((2,0),4,\infty)$};
\draw[very thick, dashed] (5,-0.8) -- (5,9.3);
\draw[very thick, dashed] (-0.8,5) -- (9.9,5);
\node at (-1.7,5) {\mbox{\scriptsize $\maxConstant$}};
\node at (5,-1.2) {\mbox{\scriptsize $\maxConstant$}};
\node at (-.8,10.0) {\scriptsize \rotatebox{0}{$y$}};
\node at (10.0,-.5) {\scriptsize $x$};
\draw[dashed] (3,0) -- (11.5,8.5);
\draw[dashed] (2,0) -- (11.5,9.5);
\filldraw[fill=black, draw=black] (1.8,-0.2) -- (2.2,-0.2) -- (2.2,0.2) -- (1.8,0.2) -- (1.8,-0.2);
\filldraw[fill=black, draw=black] (2.8,-0.2) -- (3.2,-0.2) -- (3.2,0.2) -- (2.8,0.2) -- (2.8,-0.2);
\node at (3,-0.9) {$\anchor$};
\node at (2,-0.85) {$\anchor'$};
\node at (-0.9,-0.6) {\mbox{\large $\ell'$}};

\filldraw[fill=white,draw=black, thick] (6,4) circle(6pt);
\filldraw[fill=white,draw=black, thick] (7,5) circle(6pt);
\filldraw[fill=white,draw=black, thick] (8,6) circle(6pt);
\filldraw[fill=white,draw=black, thick] (9,7) circle(6pt);
\filldraw[fill=white,draw=black, thick] (10,8) circle(6pt);
\filldraw[fill=white,draw=black, thick] (11,9) circle(6pt);
\end{scope}
\end{tikzpicture}
\label{fig:succ2}
} }

\subfloat[Reset of a clock]{
\scalebox{1.05}{
\begin{tikzpicture}[scale=0.35]
\filldraw[fill=gray!75, draw=white, thick]
    (0,0) -- (0,6) -- (9,6) -- (9,0);

\draw[help lines] (0,0) grid (9,9);
\draw[->,thick] (0,0) -- (10,0);
\draw[->,thick] (0,0) -- (0,10);
\node at (4.5,10) {$((2,0),2,5)$};
\draw[very thick, dashed] (8,-0.8) -- (8,9.9);
\draw[very thick, dashed] (-0.8,8) -- (9.9,8);
\node at (-1.7,8) {\mbox{\scriptsize $\maxConstant$}};
\node at (8,-1.2) {\mbox{\scriptsize $\maxConstant$}};
\node at (-.8,10.0) {\scriptsize \rotatebox{0}{$y$}};
\node at (10.0,-.5) {\scriptsize $x$};
\draw[dashed] (2,0) -- (11.5,9.5);
\filldraw[fill=black, draw=black] (1.8,-0.2) -- (2.2,-0.2) -- (2.2,0.2) -- (1.8,0.2) -- (1.8,-0.2);
\node at (2,-0.9) {$\anchor$};
\node at (4,-0.6) {\scriptsize $4$};
\node at (5,-0.6) {\scriptsize $5$};
\node at (6,-0.6) {\scriptsize $6$};
\node at (-0.6,6) {\scriptsize $6$};
\node at (-0.9,-0.7) {\mbox{\large $\ell$}};

\node at (4,-2.2) {\mbox{\scriptsize $\mathit{start}=2$, $\mathit{end}=6$}};
\draw[thick] (0,6) -- (9,6);

\filldraw[fill=white,draw=black, thick] (4,2) circle(6pt);
\filldraw[fill=white,draw=black, thick] (5,3) circle(6pt);
\filldraw[fill=white,draw=black, thick] (6,4) circle(6pt);
\filldraw[black] (7,5) circle(6pt);
\filldraw[black] (8,6) circle(6pt);
\filldraw[black] (9,7) circle(6pt);
\filldraw[black] (10,8) circle(6pt);
\filldraw[black] (11,9) circle(6pt);

\begin{scope}[xshift=9.2cm,yshift=-2cm]
\draw[->,thick] (1,5) to node[above] 
     {\scriptsize $y \leq 6$, $y$:=$0$} (9.1,5);
\end{scope}

\begin{scope}[xshift=19cm]
\draw[help lines] (0,0) grid (9,9);
\draw[->,thick] (0,0) -- (10,0);
\draw[->,thick] (0,0) -- (0,10);
\node at (4.5,10.9) {\scriptsize $((4,0),0,\infty)$};
\node at (4.5,10.2) {\scriptsize $((5,0),0,\infty)$};
\node at (4.5,9.5) {\scriptsize $((6,0),0,\infty)$};
\draw[very thick, dashed] (8,-0.8) -- (8,9.3);
\draw[very thick, dashed] (-0.8,8) -- (9.9,8);
\node at (-1.7,8) {\mbox{\scriptsize $\maxConstant$}};
\node at (8,-1.2) {\mbox{\scriptsize $\maxConstant$}};
\node at (-.8,10.0) {\scriptsize \rotatebox{0}{$y$}};
\node at (10.0,-.5) {\scriptsize $x$};
\draw[dashed] (4,0) -- (10.5,6.5);
\draw[dashed] (5,0) -- (10.5,5.5);
\draw[dashed] (6,0) -- (10.5,4.5);
\node at (4,-0.6) {\scriptsize $4$};
\node at (5,-0.6) {\scriptsize $5$};
\node at (6,-0.6) {\scriptsize $6$};
\node at (-0.9,-0.6) {\mbox{\large $\ell'$}};

\filldraw[fill=white,draw=black, thick] (4,0) circle(6pt);
\filldraw[fill=white,draw=black, thick] (5,1) circle(6pt);
\filldraw[fill=white,draw=black, thick] (6,2) circle(6pt);
\filldraw[fill=white,draw=black, thick] (7,3) circle(6pt);
\filldraw[fill=white,draw=black, thick] (8,4) circle(6pt);
\filldraw[fill=white,draw=black, thick] (9,5) circle(6pt);
\filldraw[fill=white,draw=black, thick] (10,6) circle(6pt);

\filldraw[fill=white,draw=black, thick] (5,0) circle(6pt);
\filldraw[fill=white,draw=black, thick] (6,1) circle(6pt);
\filldraw[fill=white,draw=black, thick] (7,2) circle(6pt);
\filldraw[fill=white,draw=black, thick] (8,3) circle(6pt);
\filldraw[fill=white,draw=black, thick] (9,4) circle(6pt);
\filldraw[fill=white,draw=black, thick] (10,5) circle(6pt);

\filldraw[fill=white,draw=black, thick] (6,0) circle(6pt);
\filldraw[fill=white,draw=black, thick] (7,1) circle(6pt);
\filldraw[fill=white,draw=black, thick] (8,2) circle(6pt);
\filldraw[fill=white,draw=black, thick] (9,3) circle(6pt);
\filldraw[fill=white,draw=black, thick] (10,4) circle(6pt);
\end{scope}

\end{tikzpicture}
\label{fig:succ3}
}}
\end{center}

\caption{Successor generation for a selected time-dart}
\label{fig:successor}
\end{figure}

%% file: alg-example.tex
\begin{figure}[t!]
\begin{center}
\begin{tikzpicture}[font=\scriptsize,yscale=1.6,xscale=1.6]
\pgfnodecircle{l1}[stroke]{\pgfxy(0,0)}{7pt}{7pt}
\pgfnodecircle{l2}[stroke]{\pgfxy(2,0)}{7pt}{7pt}
\pgfnodecircle{l3}[stroke]{\pgfxy(4,0)}{7pt}{7pt}
\pgfnodecircle{l4}[stroke]{\pgfxy(6,0)}{7pt}{7pt}

\pgfputat{\pgfxy(0,0)}{\pgfbox[center,center]{$\ell_0$}}
\pgfputat{\pgfxy(3.2,0)}{\pgfbox[center,center]{$\ell_1$}}
\pgfputat{\pgfxy(6.4,0)}{\pgfbox[center,center]{$\ell_2$}}
\pgfputat{\pgfxy(9.6,0)}{\pgfbox[center,center]{$\ell_3$}}

\pgfsetendarrow{\pgfarrowsingle}
\pgfnodeconnline{l1}{l2}
\pgfnodeconnline{l3}{l2}
\pgfnodeconnline{l3}{l4}

\pgfnodeconncurve{l2}{l2}{90}{135}{1cm}{1cm}
\pgfnodeconncurve{l2}{l3}{45}{135}{0.5cm}{0.5cm}

\pgfputat{\pgflabel{.5}{\pgfxy(1.7,1.2)}{\pgfxy(0.6,-1)}{-4pt}}{\pgfbox[center,base]{
$x \geq 2$
}}

\pgfputat{\pgflabel{.5}{\pgfxy(2.7,3.3)}{\pgfxy(1.4,-1.6)}{-4pt}}{\pgfbox[center,base]{
$x := 0$
}}

\pgfputat{\pgflabel{.5}{\pgfxy(9.5,2.5)}{\pgfxy(1.2,-1)}{-4pt}}{\pgfbox[center,base]{
$x\geq 2, y \geq 2$  $x := y:= 0$
}}

\pgfputat{\pgflabel{.5}{\pgfxy(10.3,1)}{\pgfxy(0.6,-1)}{-4pt}}{\pgfbox[center,base]{
$x \geq 1$
}}

\pgfputat{\pgflabel{.5}{\pgfxy(15,1)}{\pgfxy(0.6,-1)}{-4pt}}{\pgfbox[center,base]{
$x \leq 1, y \geq 2$
}}
\end{tikzpicture}
\end{center}

\bigskip

\begin{center}
\begin{tabular}{|l|l|r|r|r|r|r|r|r|r|}
\hline
Location & Anchor & $ 0 $  & 1 & 2 & 3 & 4 & 5 & 6 & 7 \\ \hline
$\ell_0$ & $(0,0)$  & $\mathbf{(0, \infty)}$ & $(0,0)$ & $(0,0)$ & $(0,0)$ & $(0,0)$ & $(0,0)$ & $(0,0)$ & $(0,0)$ \\ \hline
$\ell_1$ & $(0,0)$  & $\bot$ & $\mathbf{(2,\infty)}$ & $(2,2)$ & $(2,2)$ & $(2,2)$ & $\mathbf{(1,2)}$ & $(1,1)$ & $(1,1)$  \\ \hline
$\ell_1$ & $(0,1)$  & $\bot$ & $\bot$ & $\bot$ & $\bot$ & $\bot$ & $\bot$ & $\mathbf{(0,\infty)}$ & $(0,0)$ \\ \hline
$\ell_1$ & $(0,2)$  & $\bot$ & $\bot$ & $\mathbf{(0,\infty)}$ & $(0,0)$ & $(0,0)$ & $(0,0)$ & $(0,0)$ & $(0,0)$ \\ \hline
$\ell_1$ & $(0,3)$  & $\bot$ & $\bot$ & $(0,\infty)$ & $\mathbf{(0,\infty)}$ & $(0,0)$ & $(0,0)$ & $(0,0)$ & $(0,0)$  \\ \hline
$\ell_2$ & $(0,0)$  & $\bot$ & $\bot$ & $(0,\infty)$ & $(0,\infty)$ & $\mathbf{(0,\infty)}$  & $(0,0)$ & $(0,0)$ & $(0,0)$  \\ \hline
\end{tabular}
\end{center}
\caption{Example of an execution of the algorithm (columns represent the
number of iterations of the main while-loop; all unlisted pairs
of locations and anchor points are constantly having the value $\bot$)}
\label{fig:example}
\end{figure}

Let us now demonstrate the execution of Algorithm~\ref{alg:reachability}
on the automaton depicted in Figure~\ref{fig:example}, 
where we ask if the goal location
$\ell_3$ is reachable from the initial state $(\ell_0, v_0)$
where $v_0(x)=v_0(y)=0$.
The values stored in the passed-waiting list after each
iteration of the while-loop are shown in the table such that
a column labelled with a number $i$ is the status of the
passed-waiting list after the $i$'th execution of the body of the while-loop;
all values for anchor points not listed in the table are constantly $\bot$.

Initially we set $\pwList(\ell_0,(0,0))=(0,\infty)$, meaning that
all points reachable from the initial valuation after an arbitrary delay 
action belong to the waiting list and should be explored.
As $\ell_0$ is not the goal state, the algorithm continues with the
execution of the main while-loop.
In the first iteration of the loop we pick the only element in the waiting
list so that $\ell = \ell_0$, $\anchor = (0,0)$, $w=0$ and $p=\infty$.
Then we update $\pwList(\ell_0,(0,0))$ to $(0,0)$\footnote{In 
each column we mark by bold font the element that is picked in the 
next iteration of the while-loop.}
according to line~\ref{alg:addSelf} of the algorithm, meaning that all 
points on the dart are now in the passed list.
After this we consider the transition from $\ell_0$ to
$\ell_1$ with the guard $x \in [2,\infty)$ (and the implicit guard $y \in [0,\infty)$)
and calculate the values of
$\mathit{start}$ (minimum delay from the anchor point to satisfy the guard
and at the same time having at least the delay $w$ where the waiting list
starts) 
and $\mathit{end}$ (maximum delay from the anchor point so that the guard is
still satisfied). In our example we have 
$\mathit{start} = \maxOf{0, (2-0),(0-0)} = 2$ and 
$\mathit{end} = \minOf{(\infty-2), (\infty-0)} = \infty$. 

Next we consider the test at line~\ref{alg:guardsat} that requires that
the minimum delay $\mathit{start}$ to enable all guards is not
in the region of already passed points ($\mathit{start} < p$) and
at the same time that it
is below the maximum delay after which the guard become
disabled ($\mathit{start} \leq \mathit{end}$). If this test fails, there
is no need to do anything with the currently picked element from the
waiting list. 
As the values in our example satisfy the condition at line~\ref{alg:guardsat}
and no clocks are reset, we update according to line~\ref{alg:addOneSuccToPW}
of the algorithm the value of $(\ell_1, (0,0))$ to $(2,\infty)$.
This means that in the future iterations we have to explore in location $\ell_1$
all points $(2,2), (3,3), (4,4), \ldots$.
Note that the addition and subtraction of $\mathit{start}$ at 
line~\ref{alg:addOneSuccToPW} had no effect as none of the clocks after
the minimum delay exceeded the maximum constant $2$; should this happen
the values exceeding the maximum constant get truncated to $\maxConstant+1$.

In the second iteration of the while-loop 
we select the location and anchor point $(\ell_1, (0,0))$,
with $w=2$ and $p=\infty$, set it to $(2,2)$ in the table
and mark it in bold as the selected point in the previous column. 
This time we have to explore two edges.
First, we select the self-loop that resets the clock $x$ and
we get $\mathit{start}=2$ and $\emph{end}=\infty$.
Now we execute the
lines~\ref{alg:resetStart} to \ref{alg:resetStop} as the edge contains a
reset. The for-loop will be run for the value of $n$ from $2$ to $3$.
The upper-bound of $3$ for the for-loop follows from the fact that
$\maxConstant=2$ and the maximum value of a clock that is not reset 
in the anchor point is $0$.
In the for-loop we add two successors (line \ref{alg:addManySuccToPW})
at the location $\ell_1$ with the anchor points $(0,2)$ and $(0,3)$.
Second, if we consider the edge from $\ell_1$ to $\ell_2$ we can
see that in location $\ell_2$ the anchor point $(0,0)$ is
set to $(0,\infty)$.

The remaining values stored in the passed-waiting list 
are computed in the outlined way. We can notice that after the 7th 
iteration of the while-loop the set $\mathit{Waiting}(\pwList)$ is
empty and the algorithm terminates. As the location $\ell_3$ has not
been discovered during the search, the algorithm returns false.

%% file: experiments.tex
\section{Experiments}

We have
conducted a number of experiments in order to 
test the performance of the time-dart state-space representation.
The experiments were done
within the project 
opaal~\cite{DHJLOOS:NFM:11}, a model-checking framework designed 
explicitly for fast prototyping and testing of verification algorithms
using the programming language Python.
The tool implements the pseudocode of both the fully discrete (called naive
in the tables) 
as well as the time-dart reachability algorithms 
based on passed-waiting list presented in Section~\ref{sec:darts}.

The experiments were conducted on 
Intel Core 2 Duo P8600@2.4Ghz 
running Ubuntu linux. The verification was interrupted after 
five minutes or when the memory limit of 2GB RAM was exceeded 
(marked in the tables as OOM). The number of discovered symbolic states 
corresponds 
to the total number of calls to the function $\addtopw$ (including
duplicates) and the number of stored states is the size of the passed-waiting 
list at the termination of the algorithm. Verification times 
(in seconds) are highlighted in the bold font. The examples and
tool implementation are available at 
\url{http://people.cs.aau.dk/~kyrke/download/timedart/timedart.tar.gz}.

\input{table-exp-tgs.tex}
\subsection{Task Graph Scheduling}
The task graph scheduling problem (TGS) is the problem of finding 
a feasible schedule for a number of parallel tasks with given precedence 
constraints and processing times on a fixed number of homogeneous 
processors~\cite{KA:JPDC:99}. 
The chosen task graphs for two processors were taken from
the benchmark~\cite{STG}
such that several scheduling problems with different degree of 
concurrency are included. The models are scaled by the number of tasks 
in the order given by the benchmark and the verification query performed
a full state-space search.
The experimental results 
are displayed in Figure~\ref{fig:exp:tgs}.
The data confirm that the time-dart verification technique saves
both the number of stored/discovered states and noticeably improves the 
verification speed, in particular in the model T155.

\subsection{Bridge Crossing Vikings}
The bridge crossing Vikings is a slightly modified version of the 
standard planning problem available in the official distribution of \uppaal; 
we only eliminated the used integer variables that are not supported
in our opaal implementation and are simulated by new locations.
The query searched the whole state-space.
\input{table-exp-viking.tex}
Verification results are given in Figure~\ref{fig:viking}.
The performance of the time-dart algorithm is again better than 
the full discretization, even though in this case the constants
in the model are relatively small (proportional to the number
of Vikings), meaning that the potential
of time-darts is not fully exploited. 


\subsection{Train Level Crossing} 
\input{table-exp-train.tex}
In train level crossing we consider auto-generated timed automata
templates constructed via automatic translation~\cite{BJS:ICFEM:09} 
from timed-arc Petri net model of a train level-crossing example. 
The auto-generated timed automata were produced by the tool 
\tapaal~\cite{DJJJMS:TACAS:12} 
and have a rather complex structure that human modelers
normally never design and hence we can test the potential
of the discrete-time engine also for the models translated from
other time-dependent formalisms.  
The query we asked searches the whole state-space.
We list the results in Figures~\ref{fig:train}
and the experiment demonstrates again the advantage
of the time-dart verification method.

\subsection{Fischer's Protocol}
\input{table-exp-fischer.tex}
The discrete-time techniques
are sensitive to the size of the constants present in the model.
We have therefore scaled our next experiment by the size of the maximal
constant (MC) that appears in the model in order to demonstrate the main
advantage of the time-dart algorithm. For this we use the well known
Fischer's protocol for ensuring a mutual 
exclusion between two or more parallel processes \cite{lamport87fast}. 
It is a standard model for testing the performance of verification tools;
we replaced one open interval in the model with a closed one such that
mutual exclusion is still guaranteed.
The concrete version of the protocol we verified was created
by a translation from timed-arc Petri net model of the protocol~\cite{AN:BQQ}
available as a demo example in the tool \tapaal\cite{DJJJMS:TACAS:12}. 
We searched the whole state-space and the results are summarized in 
Figure~\ref{fig:fischer}. It is clear that time-darts are superior
w.r.t. the scaling of the constants in the model, allowing us to verify 
(within the given limit of 300 seconds) models where the maximum constant 
is 66, opposed to only 18 when the full discretization is used.



%% file: table-exp-tgs.tex
\begin{figure}[t!]
{

\center
\small
\begin{tabular}{|lr|>{\bfseries}r|r|r|>{\bfseries}r|r|r|}
\hline
Model&\#
&\textbf{Naive}&Discovered&Stored&\textbf{Darts}&Discovered&Stored\\
\hline
\hline
T55&4&3.8&81,062&38,906&1.9&34,012&3,654\\
\hline
T55&5&13.0&254,969&110,907&5.9&111,543&10,739\\
\hline
T55&6&43.8&727,712&297,026&17.6&336,527&29,378\\
\hline
T55&7&95.7&1,431,665&524,270&32.0&607,483&51,730\\
\hline
T55&8&OOM&&&91.8&1,740,066&136,639\\
\hline
T55&9&-&-&-&255.7&4,700,607&347,136\\
\hline
T55&10&-&-&-& $>$300&-&-\\
\hline
\hline
T125&4&0.3&2,609&2,050&0.2&198&139\\
\hline
T125&5&1.6&18,394&14,772&0.2&713&503\\
\hline
T125&6&5.5&61,242&48,600&0.5&2,916&1,769\\
\hline
T125&7&20.3&205,808&161,394&1.4&10,337&6,102\\
\hline
T125&8&93.4&82,4630&529,032&4.7&39,242&20,392\\
\hline
T125&9&OOM&-&-&13.7&111,438&56,191\\
\hline
T125&10&-&-&-&34.5&274,939&126,895\\
\hline
T125&11&-&-&-&OOM&-&-\\
\hline
\hline
T155&4&0.4&5,796&3,048&0.3&1,532&467\\
\hline
T155&5&1.1&23,454&9,740&0.4&4,572&1,195\\
\hline
T155&6&19.9&433,674&14,2861&3.6&62,771&8,859\\
\hline
T155&7&28.5&577,179&18,7857&4.5&74,105&10,725\\
\hline
T155&8&32.3&620,138&203,178&4.8&78,093&11,508\\
\hline
T155&9&34.1&626,100&205,646&4.9&79,111&11,753\\
\hline
T155&10&60.8&1,035,226&329,193&7.5&241,44&18,574\\
\hline
T155&11&OOM&-&-&31.2&514,959&67,592\\
\hline
T155&12&-&-&-&37.7&608,974&80,634\\
\hline
T155&13&-&-&-&105.7&1,684,525&205,087\\
\hline
T155&14&-&&-&158.3&2,316,474&284,859\\
\hline
T155&15&-&-&-&164.5&2,409,417&298,288\\
\hline
T155&16&-&-&-&OOM&-&- \\
\hline

\end{tabular}
\caption{Results for three different TGS scaled by the number of tasks}
\label{fig:exp:tgs}

}
\end{figure}

%% file: table-exp-viking.tex
\begin{figure}[t]
\center
\small
\begin{tabular}{|l|>{\bfseries}r|r|r|>{\bfseries}r|r|r|}

\hline
\#&\textbf{Naive}&Discovered&Stored&\textbf{Darts}&Discovered&Stored \\
\hline
\hline
2&0.2&295&152&0.1&87&46\\
\hline
3&0.2&2,614&1,263&0.2&754&336\\
\hline
4&0.8&16,114&7,588&0.5&4,902&1,759\\
\hline
5&4.9&90,743&42,294&2.5&29,144&8,308\\
\hline
6&33.2&501,958&235,635&13.4&165,535&38,367\\
\hline
7&OOM&-&-&74.8&900,439&177,807\\
\hline
8&-&-&-&$>$300&-&-\\
\hline
\end{tabular}

\caption{Results for bridge crossing scaled by the number of Vikings}
\label{fig:viking}
\end{figure}

%% file: table-exp-train.tex
\begin{figure}[t]
\small
\center
\begin{tabular}{|l|>{\bfseries}r|r|r|>{\bfseries}r|r|r|}

\hline
\#&\textbf{Naive}&Discovered&Stored&\textbf{Darts}&Discovered&Stored \\

\hline
\hline
1 &0.2&173&139&0.2&21&15\\
\hline
2&1.5&16684&11042&0.3&1140&647\\
\hline
3&OOM&-&-&4.6&40671&21721\\
\hline
4&- &-&-& OOM &-&-\\
\hline
\end{tabular}

\caption{Results for train level crossing scaled by the number of trains}
\label{fig:train}
\end{figure}

%% file: table-exp-fischer.tex
\begin{figure}[t]
\small
\center
\begin{tabular}{|l|>{\bfseries}r|r|r|>{\bfseries}r|r|r|}

\hline
MC&\textbf{Naive}&Discovered&Stored&\textbf{Darts}&Discovered&Stored \\
\hline
\hline
3&3.1&36,774&25,882&1.1&9,464&6,238\\
\hline
4&4.5&53,655&38,570&1.4&14,341&8,725\\
\hline
5&6.2&74,415&54,513&1.8&20,226&11,548\\
\hline
9&17.8&202,965&157,555&3.8&53,846&26,200\\
\hline

15&64.5&569,280&466,888&10.4&133,796&58,258\\
\hline

18&111.8&850,164&710,857&14.2&187,595&78,823\\
\hline
19& OOM & - & - & 15.7& 207,544&86,350  \\
\hline
25&-&-&-&26.1&348,406&138,568 \\
\hline
38&-&-&-&61.3&778,095&293,203 \\
\hline
50&-&-&-&114.7&1,325,931&486,343 \\
\hline
66&-&-&-&217.2&2,214,846&795,808 \\
\hline
\end{tabular}

\caption{Experimental results for Fischer's mutual exclusion protocol}
\label{fig:fischer}
\end{figure}

%% file: conclusion.tex
\section{Conclusion}
We have introduced a new data structure of time-darts
in order to represent the reachable state-space of closed timed automata
models.
We showed on a number of experiments that our time-dart reachability algorithm
achieves a consistently better performance than the explicit search algorithm,
improving both the speed and memory requirements. This is obvious in particular
on models with larger constants (as demonstrated in the Fischer's experiment
or T155 task graph) 
where time-darts provide a compact representation of the delay successors and 
considerably improve both time and memory.

The algorithms were implemented in the interpreted language Python
without any further optimizations techniques like
partial order and symmetry reductions and advanced extrapolation techniques
and with only one global maximum constant.
This does not allow us to compare its performance directly with
the state-of-the-art optimized tools for real-time systems. 

An advantage of time-darts and explicit state-space methods in general is
that it is relatively easy to extend them with additional modelling
features like clock invariants and diagonal guards.
In our future work we will implement the time-dart algorithm in {\tt C++}
with additional optimizations (e.g. considering local constants instead of the 
global ones) and we shall also consider the verification of liveness properties.
It is clear that for large enough constants the DBM-search engine will always
combat the explicit methods (see~\cite{DBLP:conf/charme/Lamport05}); 
our technique can be so seen as a practical alternative
to the DBM-engines on the subset of models that for example use
counting features (like in our introductory example) and where DBM state-space
representation explodes even for models with small constants.
Another line of research will focus on further optimizations of the
time-dart technique by considering federations of time-darts
so that the data structure becomes even less sensitive to
the scaling  of the constants. 

\paragraph{Acknowledgements.} We would like to thank the anonymous
reviewers for their comments.